\documentclass{llncs}

\usepackage[latin1]{inputenc}
\usepackage{subfigure}

\usepackage{graphicx}

\usepackage[round]{natbib}

\usepackage{amsmath}
\usepackage{amssymb}
\usepackage{amsfonts}
\usepackage {graphicx}
\usepackage[latin1]{inputenc}
\usepackage[ruled,vlined,titlenumbered,linesnumbered,longend]{algorithm2e}

\usepackage{tikz}

\newcommand{\ds}{\displaystyle}

 \title{Boltzmann Samplers for $v$-balanced Colored Necklaces}
 \author{O. Bodini \and A. Jacquot }
\institute{
 LIP6, UMR 7606, Departement CALSCI,
 Universit\'e Paris 6 - UPMC,\\
 104, avenue du pr\'esident Kennedy,\\ 
 F-75252 Paris cedex 05, France
 }
 
\begin{document}

\maketitle

 \begin{abstract}
 This paper is devoted to the random generation of particular colored necklaces for which the number of beads of a given color is constrained (these necklaces are called \emph{$v$-balanced}). We propose an efficient sampler (its expected time complexity is linear) which satisfies the Boltzmann model principle introduced by Duchon, Flajolet, Louchard and Schaeffer \cite{DFLS}. Our main motivation is to show that the absence of a decomposable specification can be circumvented by mixing the Boltzmann samplers with other types of samplers.\\
 \end{abstract}


%

\section*{Introduction}
Necklaces are classical objects in combinatorics \cite{FS,CL,FK,FM}. They naturally occur in the study of Lyndon words or in many other enumeration problems \cite{CRMSS}. For $v=(v_1,...,v_k)$ a $k$-tuple of positive integers, our interest lies in uniformly drawing $v$-balanced necklaces of $n$ beads : the beads can take $k$ distinct colors and the number $n_i$ of  beads of color $i$ verifies the \textit{$v$-balance} (which we define as meaning that $(n_1,...,n_k)$ is collinear to $v$). An additional reason to focus on $v$-balanced structures derives from the intrinsic difficulty to enumerate and describe such objects in terms of analytic combinatorics. In particular, the generating function of $v$-balanced cycles is neither holonomic nor closed-formed.
We attempt to draw very large necklaces in order to conjecture some limit properties on them. To that purpose, we adopt the framework of Boltzmann samplers. This approach, like the recursive method \cite{NW}, allows to ``automatically'' build a sampler from a decomposable specification of a combinatorial class. A Boltzmann sampler does not guarantee the size of the generated object, only that the drawn object has the same probability to be drawn as any other of the same size. More details on this point can be found in the preliminary section. But we can already state that for a large number of combinatorial classes this relaxation allows to generate an object of size $n$ in expected linear time $O(n)$ without preprocessing. 

The main problem is that $v$-balanced necklaces do not admit decomposable s using traditional builders $(+,\times,Seq,...)$, except for some very special cases. The aim of this paper is to adapt an ad-hoc sampler for $v$-balanced sequences, in such a way that the sampling follows a Boltzmann model. After that, we will use this Boltzmann sampler (which cannot be built from a decomposable specification) to obtain a Boltzmann sampler for $v$-balanced cycles. To this end, we need to know how to obtain a Boltzmann sampler for a class $\mathcal{A}$ from a Boltzmann sampler for its pointing class $\Theta\mathcal{A}$.

\smallskip
This paper is organized in four sections.
The first section defines the notations and concepts used throughout the article. After reflecting upon basic notions related to combinatorial classes and Boltzmann samplers, we define our classes of interest, namely, the $v$-balanced sequences and the $v$-balanced cycles.\\
The second section addresses the sampler for $(1,1)$-balanced cycles. We show two ways of building such a sampler. The first one uses a natural isomorphism with a decomposable combinatorial class involving Dyck paths. Unfortunately, this approach cannot be generalized, therefore we propose a second one based on an isomorphism between pointed $(1,1)$-balanced cycles and the weighted sum of $(1,1)$-balanced sequences. These sequences can be drawn with an ad-hoc sampler.\\
In the third section, a generalization of the previous approach to $v$-balanced cycles is given, yielding an efficient Bolzmann sampler which, for our object of study, bypasses the need for a specification of the class.\\
The fourth section concludes with some perspective works.

\section{Preliminaries}
An \textit{(unlabelled) combinatorial class} is a couple $\langle \mathcal{C},s\rangle$ (generally abbreviated by  $\mathcal{C}$)  where $\mathcal{C}$ is a set of \textit{objects} and $s$ is a function on $\mathcal{C}$ called \textit{size function} which satisfies the conditions :\\
(i) $\forall a\in \mathcal{C}$, $s(a)$ is a non-negative integer;\\
(ii) the number of objects of any given size is finite.\\
We naturally associate to a combinatorial class $\mathcal{C}$ the ordinary generating function $C(x)=\sum c_nx^n$ where $c_n$ is the number of objects of size $n$ in $\mathcal{C}$.
A \textit{Boltzmann sampler} for an unlabelled combinatorial class $\mathcal{C}$ is a random generator such that the probability to draw an object $a$ of size $n$ is $P_x(a)=\dfrac{x^n}{C(x)}$. These samplers were first introduced in \cite{DFLS} and extended, (in particular the sampler for unlabelled cycles) in \cite{FFP}. Let us notice that a Boltzmann sampler depends on a parameter $x$ which can be tuned to focus on an expected output size. More precisely, let $N$ be the random variable of the size of the output, we can solve the equation $\mathbb{E}(N)=x\cdot\dfrac{C'(x)}{C(x)}$ to center the output distribution on an expected value. 

Let us recall in the following table the samplers for the unlabelled operators that we need. 

\bigskip
\begin{figure}[htbp]
\begin{center}

\noindent
\begin{tabular}{|l|l|}
\hline
Sampler& Description\\
\hline
\hline
$\Gamma_x(\mathcal{Z})$ &\textbf{Return} $\mathcal{Z}$. \qquad $\mathcal{Z}$ denotes an atomic class.\\

\hline
$\Gamma_x (\mathcal{A}\mathcal{B})$& \textbf{Return} ($\Gamma_x\mathcal{A}, \Gamma_x\mathcal{B}$).\\
\hline
& If $Bernoulli(\frac{A(x)}{A(x)+B(x)})=1$\\
$\Gamma_x (\mathcal{A}+\mathcal{B})$&then \textbf{Return} $\Gamma_x\mathcal{A}$\\
&else \textbf{Return} $\Gamma_x\mathcal{B}$\\
\hline
$\Gamma_x Seq(\mathcal{C})$&Draw $l$ according to a geometric law of parameter $C(x)$.\\
&\textbf{Return} the concatenation of $l$ calls to $\Gamma_x\mathcal{A}$.\\
\hline
$\Gamma_x Rep_n(\mathcal{C}) $& \textbf{Return} $\underbrace{aa...a}_{n \mathrm{\;times}}$ with $a$ generated by $\Gamma_{x^n} \mathcal{A}$.\\
\hline
&Let $K$ be a random variable in $\mathbb{N}^*$ verifying :\\
&$\mathbb{P}(K=k)=\frac{-\varphi(k)}{kCycA(x)}\log(1-A(x^k))$ with $CycA(x)$ the generating function of~$Cyc(\mathcal{A})$\\
$\Gamma_x Cyc(\mathcal{C})$&Draw $k$ according to the law of $K$.\\ &Draw $j$ according to a logarithmic law of parameter $A(x^k)$.\\
&Let $M$ by the concatenation of $j$ calls to $\Gamma_{x^k}\mathcal{A}$ \\
&\textbf{Return}  $\underbrace{MM...M}_{k \mathrm{\;times}}$.\\
\hline
\end{tabular}

\end{center}
\caption{Some classical samplers with  $\mathcal{A}$, $\mathcal{B}$ and $\mathcal{C}$ combinatorial classes, $A(x)$, $B(x)$ and $C(x)$ their respective generating functions.  $\mathcal{C}$ must not have neutral objects.}
\end{figure}

\bigskip
Boltzmann samplers are very powerful tools to efficiently generate combinatorial structures \cite{BFP,DS}. 
In particular, It can be possible to automatically build a sampler according to the specification of a combinatorial class. Our aim is to consider $Cyc_v$ and
$Seq_v$ defined below, as additional basic classes, to be added to the collection of classical constructions thus increasing the expressivity of the Boltzmann model. 

\begin{definition}
 We denote by $Seq_{v}$ (resp. $Cyc_{v}$) a sequence of atoms (objects of size 1), such that each atom can be colored by one of the color in $\{1,...,k\}$ and the numbers $n_i$ of beads of color $i$ verifies the $v$-balanced condition. That is to say $(n_1,...,n_k)$ is collinear to $v$.
\end{definition}

 From an easy observation, it can be seen that the generating function of $Seq_{(v_1,...,v_k)}$ is $\displaystyle\sum (\dfrac{(n|v|)!}{\prod((v_in)!)})x^{n|v|}$  where $|v|~=~\sum v_i$. The following proposition is a trivial consequence.

\begin{proposition}[see e.g \cite{BH}]
 The class $Seq_{v}$ is holonomic for every $v$.
\end{proposition}
\begin{proof} Indeed, a quick calculation shows that this can be written as the hypergeometric function : $$_kF_{k-1}(\frac{1}{|v|},\frac{2}{|v|},...,\frac{|v|}{|v|};\frac{1}{v_1},\frac{2}{v_1},...,\frac{v_1}{v_1},\frac{1}{v_2},...,\frac{v_2}{v_2},...,\frac{1}{v_k},\frac{2}{v_k},...,\frac{v_k-1}{v_k};\frac{|v|^{|v|}}{\prod v_i^{v_i}}\times x^{|v|}).$$ The hypergeometric functions are holonomic by definition.
\end{proof}

\noindent\textit{Remark:}
 Nevertheless, the class $Seq_{v}$ is algebraic in only very few cases. For instance, $Seq_{(1,1,1)}$ is not algebraic. See \cite{BH} for a complete classification of the algebraic cases.


\section{Generating (1,1)-balanced cycles}

This section is dedicated to the random generation of $(1,1)$-balanced cycles. This is a good example to illustrate our approach. 
A $(1,1)$-balanced sequence (resp. cycle) is by previous definition a sequence (resp. cycle) of black atoms $\mathcal{Z}_b$ and white atoms $\mathcal{Z}_w$ such that the number of black atoms is equal to the number of white atoms. 
The generating function for $Seq_{(1,1)}$ is $S_{(1,1)}= \,_1F_0(1/2,4x^2)=(1-4x^2)^{-\frac{1}{2}}.$ In particular, $Seq_{(1,1)}$ is algebraic. We are going to use this important property in the following first approach.
\subsection{First approach: through a decomposable specification}


Both $Seq_{(1,1)}$ and $Cyc_{(1,1)}$ are specifiable. We can thus apply the unlabelled samplers described in the previous table.
More precisely, $Seq_{(1,1)}$ is a classical combinatorial notion, names \textit{bridges}.
To generate $Cyc_{(1,1)}$, we use an isomorphism between the  $(1,1)$-balanced cycles and the cycles of indecomposable Dyck paths. 

Let $\mathcal{D}$ be the class  of Dyck paths of specification $\mathcal{D}=Seq(\nearrow \mathcal{D}\searrow)$. Dyck paths are excursions from $(0,0)$ to $(0,2n)$ over the discrete lattice $\mathbb{Z}^+\times\mathbb{Z}^+$, with displacements of $(1,1)$ and $(1,-1)$. This can also be viewed as the class of well-formed parentheses strings.
Bridges are defined by  $Seq((\nearrow \mathcal{D}\searrow)+(\searrow \overline{\mathcal{D}}\nearrow))$, where $\overline{\mathcal{D}}$ is like $\mathcal{D}$, but with the roles of $\searrow$ and $\nearrow$ interchanged. So, we can generate $Seq_{(1,1)}$ by classical Boltzmann sampler principles. 
%

The class of \textit{indecomposable} Dyck paths is the class of specification $\nearrow \mathcal{D}\searrow$.

\begin{proposition}
[Raney's lemma] The balanced cycle can also be specified as $Cyc_{(1,1)}\simeq Cyc(\nearrow \mathcal{D}\searrow)$, where $Cyc$ is the classical constructor for cycles.\\
\end{proposition}

\begin{proof}(sketch)
We can represent $Cyc_{(1,1)}$ as an excursion from $(0,0)$ to $(0,2n)$, with steps of $(1,1)$ and $(1,-1)$, up to circular permutations. Such an excursion has a non-empty set $S$ of points of minimal abscissa. 
As we deal with excursions up to circular permutations, we can consider that the excursion begin at any $c$ in $S$ (see Fig.2).
In this case, we have Dyck paths and we remark that they are exactly the same up to circular permutations of their indecomposable Dyck paths.

\begin{figure}
  \begin{center}
\noindent$Cycl(\circ\:\!\bullet\:\!\bullet\:\!\bullet\:\!\circ\:\!\bullet\:\!\circ\:\circ ) $
$\leftrightarrow$
\begin{tikzpicture}[baseline]
\draw (0,0) -- (0.2,0.2)--(0.4,0)--(0.6,-0.2)--(0.8,-0.4)--(1.0,-0.2)--(1.2,-0.4)--(1.4,-.2)--(1.6,0);
\fill (0,0) circle(0.4mm);
\fill (0.2,0.2) circle(0.2mm);
\fill (0.4,0) circle(0.2mm);
\fill (.6,-0.2) circle(0.2mm);
\fill (1.0,-0.2) circle(0.2mm);
\fill (1.2,-0.4) circle(0.2mm);
\fill (1.4,-.2) circle(0.2mm);
\fill (1.6,0) circle(0.4mm);
\fill (0.8,-0.4) circle(0.5mm);
\draw[very thin,dashed  ] (0,0) -- (1.6,0);
\draw (0,0)[anchor=north] node {a};
\draw (1.6,0)[anchor=north] node {a};
\draw (0.8,-0.4)[anchor=north] node {c};
\end{tikzpicture}
$\equiv$
\begin{tikzpicture}[baseline]
\draw (0,0) -- (0.2,0.2)--(0.4,0)--(0.6,0.2)--(0.8,0.4)--(1,.6)--(1.2,0.4)--(1.4,.2)--(1.6,0);
\fill (0,0) circle(0.5mm);
\fill (0.2,0.2) circle(0.2mm);
\fill (0.4,0) circle(0.2mm);
\fill (0.6,0.2) circle(0.2mm);
\fill (0.8,0.4) circle(0.4mm);
\fill (1,.6) circle(0.2mm);
\fill (1.2,0.4) circle(0.2mm);
\fill (1.4,.2) circle(0.2mm);
\fill (1.6,0) circle(0.5mm);
\draw[very thin,dashed  ] (0,0) -- (1.6,0);
\draw (0,0)[anchor=north] node {c};
\draw (1.6,0)[anchor=north] node {c};
\draw (0.8,0.4)[anchor=north] node {a};
\end{tikzpicture}
$\leftrightarrow$
$Cycl($
\begin{tikzpicture}[baseline]
\draw (0,0) -- (0.2,0.2)--(0.4,0);
\fill (0,0) circle(0.4mm);
\fill (0.2,0.2) circle(0.2mm);
\fill (0.4,0) circle(0.4mm);
\draw[very thin,dashed  ] (0,0) -- (.4,0);
\end{tikzpicture}
$\,$
\begin{tikzpicture}[baseline]
\draw (0.0,0)--(0.2,0.2)--(0.4,0.4)--(.6,.6)--(.8,0.4)--(1,.2)--(1.2,0);
\fill (0.0,0) circle(0.4mm);
\fill (0.2,0.2) circle(0.2mm);
\fill (0.4,0.4) circle(0.2mm);
\fill (0.6,.6) circle(0.2mm);
\fill (0.8,0.4) circle(0.2mm);
\fill (1,.2) circle(0.2mm);
\fill (1.2,0) circle(0.4mm);
\draw[very thin,dashed  ] (0,0) -- (1.2,0);
\end{tikzpicture}$)$
\end{center}
    \caption{Isomorphism between $Cyc_{(1,1)}$ and $Cyc(\nearrow \mathcal{D}\searrow)$.}
\end{figure}
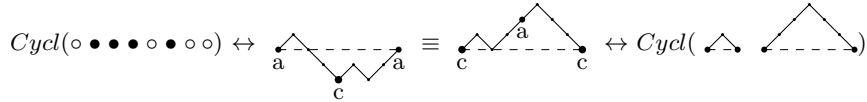

\end{proof}

Now, with this specification, we can use usual Boltzmann samplers for unlabelled structures to draw $(1,1)$-balanced cycles. We obtain a Boltzmann sampler for $Cyc(\nearrow \mathcal{D} \searrow)$ by a combination of the previous mentioned samplers according to the specification. In the following algorithms $D(x)$ and $C_{(1,1)}(x)$ respectively represent the generating functions of $\mathcal{D}$ and $Cyc(\circ D \bullet)$.

\begin{algorithm}[H]
\caption{$\Gamma_x \mathcal{D}$}

\KwIn{the parameter $x$}

\KwOut{an object of $\mathcal{D}$}
Draw $l$ according to a geometric law of parameter $x^2D(x)$\\
$M := \varepsilon$\\
\For{$i$ \emph{\textbf{from }} $1$ \textbf{\emph{to }} $l$}{
$M := concat(M, \circ, \Gamma_x \mathcal{D}, \bullet$)
}

\Return $(M)$.
\end{algorithm}
\begin{algorithm}[H]
\caption{$\Gamma_x Cyc_{(1,1)}$}

\KwIn{the parameter $x$}

\KwOut{a (1,1)-balanced cycle}

Let $K$ be a random variable in $\mathbb{N}^*$ verifying :\\
$\mathbb{P}(K=k)=\frac{-\varphi(k)}{kC_{(1,1)}(x)}\log(1-x^{2k}D(x^k))$\\
Draw $k$ according to the law of $K$.\\
Draw $j$ according to a logarithmic law of parameter $x^{2k}D(x^k)$.\\
$M := \varepsilon$\\
\For{$i$ \emph{\textbf{from }} $1$ \textbf{\emph{to }} $j$}{
$M := concat(M, \circ, \Gamma_x \mathcal{D}, \bullet)$
}
\Return $[\underbrace{MM...M}_{k \mathrm{\;times}}]$.
\end{algorithm}
 
This method is an efficient way to draw $(1,1)$-balanced cycles. In particular, the basic rejection sampler $\mu Cyc_{(1,1)}(x;n,\varepsilon)$ (see \cite{DFLS} for details) has an $O(n)$ overall cost in average. But it relies on a very singular property of the class~: it can be decomposed with usual constructors. We are now interested in another way to generate these objects. This new approach will be extended to all $v$-balanced objects in the last section.
\subsection{Second approach: mixed samplers}

We are still focused on the generation of $(1,1)$-balanced cycles, but now  the use of an algebraic specification is avoided.
The idea of our sampler can be summerized as follows : first, we adapt an ad-hoc sampler for the $(1,1)$-balanced sequences in such a way that this sampling follows a Boltzmann model; second, we show an isomorphism (see proposition 4) between the class $\Theta Cyc_{(1,1)}$ of pointed balanced cycles and a sum involving duplications of $Seq_{(1,1)}$. The notion of pointing classes is recalled in this part; finally, to obtain a Boltzmann sampler for $Cyc_{(1,1)}$, we explain how to obtain a Boltzmann sampler for a class $A$ from a Boltzmann sampler for its pointing class $\Theta A$.

\subsubsection{Sampler for $Seq_{(1,1)}$.}
Let us start by introducing our Boltzmann sampler  for $Seq_{(1,1)}$ and proving its correctness:



\begin{algorithm}[H]
\caption{$\Gamma_x Seq_{(1,1)}$}

\KwIn{the parameter $x$}

\KwOut{a balanced sequence of $\mathcal{Z}_b$ and $\mathcal{Z}_w$.}

Let $L$ be a random variable in $\mathbb{N}^*$ verifying 
$\mathbb{P}(L= l)=\frac{(2l)!}{(l!)^2} \frac{x^{2l}}{S_{(1,1)}(x)}$\\
Draw $l$ according to the law of $L$.\\
Let $M$ be a $(2l)$-uple,
select uniformly $l$ positions belongs the $2l$ entries in $M$.\\ This positions are the $\mathcal{Z}_b$ entries of $M$, the other ones are the $\mathcal{Z}_w$ entries.

\Return $(M)$.
\end{algorithm}
\begin{lemma}
 Algorithm 3 is a valid Boltzmann sampler for $Seq_{(1,1)}$.
\end{lemma}

\begin{proof}
Let $\alpha$ be an output of this algorithm.
The probability to draw $\alpha$ is the probability to draw the right length and then to draw the right positions for $\mathcal{Z}_b$. So,
$$\mathbb{P}(\alpha)= \frac{|\alpha|!}{(\frac{|\alpha|}{2}!)^2} \frac{x^{|\alpha|}}{S_{(1,1)}(x)} . \frac{(\frac{|\alpha|}{2}!)^2}{|\alpha|!} = \frac{x^{|\alpha|}}{S_{(1,1)}(x)}$$

\end{proof}

\subsubsection{An isomorphism for $\Theta Cyc_{(1,1)}$.}








Another classical operator, in structural combinatorics, is the pointing operator which can be defined as follows :

\begin{definition} Let $\mathcal{C}$ be a combinatorial class, the combinatorial class $\Theta\mathcal{C}$ is formally defined as $\displaystyle{\sum_{n>0} \mathcal{C}_n\times\{\epsilon_1,...,\epsilon_n\}}$ where the $\epsilon_i$ are distinct neutral objects (0-sized objects) and $\mathcal{C}_n$ is the sub-class of $\mathcal{C}$ of the elements of size $n$. 
\end{definition}
The generating function of the pointed class $\Theta\mathcal{C}$ is $C^{\bullet}(x)=x.\dfrac{dC}{dx}$. We can interpret this operator by saying that each object in $\Theta\mathcal{C}$ is an object in $\mathcal{C}$ with a tagged atom on it.

\begin{theorem}
 $\Theta Cyc_{(1,1)}$ is isomorphic to $\sum\limits_{n>0} \varphi(n) Rep_n(Seq_{(1,1)})$, where $Rep_n(\mathcal{A})$ is the class $\{\underbrace{aa...a}_{n \mathrm{\;times}}; a\in\mathcal{A}\}$.
\end{theorem}

\begin{proof}

Let $C_{bic}$ be the generating function for $Cyc(\mathcal{Z}_w+\mathcal{Z}_b)$, without any constraint on the number of beads ($\mathcal{Z}_w$ and $\mathcal{Z}_b$) of each color. We can write the generating function $C_{bic}$ as follows :
$$ C_{bic}=\sum\limits_{n>0}\frac{-\varphi(n)}{n}\log(1-(Z_w^n+Z_b^n))=\sum\limits_{n>0}\frac{\varphi(n)}{n} \sum\limits_{k>0}\frac{1}{k} (Z_b^n+Z_w^n)^k$$
$$C_{bic}=\sum\limits_{n>0}\frac{\varphi(n)}{n} \sum\limits_{k>0}\frac{1}{k} \sum\limits_{p_1+p_2 = k} \frac{k!}{p_1!p_2!}Z_b^{p_1n}Z_w^{p_2n} $$


The notion of diagonal for a bivariate generating function is needed for what follows. Let $f(X,Y)=\sum_i a_{i,j}X^iY^j$ be a bivariate generating function, the function $g(Z)=\sum a_{n,n} Z^{2n}$ is called the \textit{(1,1)-diagonal }of $f(X,Y)$ and  it is denoted by $\Delta f$. 

So, by definition $C_{(1,1)}=\Delta C_{bic}$ and we can use the previous formula to obtain :
$$C_{(1,1)}=\sum\limits_{l>0}([Z_b^l Z_w^l]C_{bic})X^{2l}=\sum\limits_{n>0}\frac{\varphi(n)}{n}\sum\limits_{p>0}\frac{1}{2p} \frac{(2p)!}{(p!)^2}X^{2np}$$

Now, pointing the Cyc(1,1) class yields :
$$ C_{(1,1)}^{\bullet}=\sum\limits_{n>0} \varphi(n) \sum\limits_{p>0}\frac{(2p)!}{(p!)^2}X^{2np}$$
$$ C_{(1,1)}^{\bullet}=\sum\limits_{n>0} \varphi(n) S_{(1,1)}(X^{n}) $$


\end{proof}

At this stage, it is possible to draw an object of size $n$ in $ Cyc_{(1,1)}$ using classical recursive method \cite{FZV}. But here we pitch on Boltzmann point of view which avoids costly preprocessing calculus.

This isomorphism allows us to describe a Boltzmann sampler for $\Theta Cyc_{(1,1)}$~:


\begin{algorithm}[H]
\caption{$\Gamma_x \Theta Cyc_{(1,1)}$}

\KwIn{the parameter $x$}

\KwOut{a $(1,1)-$balanced cycle of $\mathcal{Z}_b$ and $\mathcal{Z}_w$.}

$n:=1$\\
$S:=\frac{\varphi(n) S_{(1,1)}(x^{n}) } {C_{(1,1)}^{\bullet}(x)}$\\
Draw a real number $u$ uniformly in $[0,1]$\\
\While{$u>S$}{
$n:=n+1$\\
$S:=S+\frac{\varphi(n) S_{(1,1)}(x^{n}) } {C_{(1,1)}^{\bullet}(x)}$
}
\Return $[\Gamma_x (Rep_n(Seq_{(1,1)}))]$\\
\end{algorithm}
\begin{corollary}
 Algorithm 4 is a valid Boltzmann sampler for $\Theta Cyc_{(1,1)}$.
\end{corollary}
\begin{proof}
This is a corollary of the correctness of our general sampler, given in section four.
\end{proof}
\vspace{-0.7cm}
\subsubsection{A Boltzmann Sampler for $Cyc_{(1,1)}$.}

We have obtained a Boltzmann sampler for  $\Theta Cyc_{(1,1)}$. This is enough to uniformly generate (1,1)-balanced cycles. 
But, the sampler does not have a Boltzmann distribution for $Cyc_{(1,1)}$, so it can not be called by another constructor.
 For instance,  $Cyc(\Theta Cyc_{(1,1)})$ is not equal to $Cyc(Cyc_{(1,1)})$.
Indeed, small objects are drawn with a smaller probability in $\Theta Cyc_{(1,1)}$. 
To unbias the sampler, we are going to change its parameter according to a well-chosen density law $f_x(u)$. A similar idea occurs in \cite{R}.

\begin{lemma} Let $C(x)=\sum_{n>0} c_n x^n$ be a generating function (with $C(0)=0$). 
 For any fixed $x$ in the convergence disc of $C$, the function $f_x(u)=\dfrac{C^{\bullet}(xu)}{uC(x)}$ is a density of probability on $[0,1]$.
\end{lemma}
\begin{proof}
 Clearly $f$ is non negative. Now, it remains only to prove that \\ $\displaystyle\int_{u=0}^{1}{\dfrac{C^{\bullet}(xu)}{uC(x)}du}=1$. We can expand the serie to $C^{\bullet}(xu)$, $\displaystyle\int_{u=0}^{1}{\dfrac{\sum_{n>0} nc_n(ux)^n}{uC(x)}du}$. We now swap the sum and integral, we have : 
$$\displaystyle{\sum_{n>0}{\dfrac{nc_nx^n}{C(x)}\int_{u=0}^{1}{u^{n-1}}du}=\dfrac{1}{C(x)}\sum_{n>0}{nc_nx^n\left[\dfrac{u^n}{n}\right]_{u=0}^{1}}=1}.$$
\end{proof}

\begin{theorem}
 The following sampler (Algorithm 5) gives a valid Boltzmann sampler for $\mathcal{C}$ with parameter~$x$ from a Boltzmann sampler for $\Theta \mathcal{C}$.
\end{theorem}

\begin{algorithm}[H]
\caption{$\Gamma_x \mathcal{C}$}

\KwIn{the parameter $x$}

\KwOut{an object in $\mathcal{C}$.}

Draw a real number $u$ according to the density law $\dfrac{C^{\bullet}(xu)}{u(C(x)-c_0)}$. \\

\eIf{({Bernoulli}$(\dfrac{c_0}{C(x)}) = 1$)}{
\Return an object in $\mathcal{C}_0$ drawn uniformly.}{
\Return $\Gamma_{ux}(\Theta\mathcal{C})$ and forget the point.}
\end{algorithm}

\begin{proof}
 It is sufficient to evaluate the probability that the output be of size $n$. If $n=0$, we have drawn an object in $\mathcal{C}_0$. This occurs with probability $\dfrac{c_0}{C(x)}$. If $n>0$, the probability is $(1-\dfrac{c_0}{C(x)})\cdot\displaystyle\int_{u=0}^{1}{\dfrac{nc_n(ux)^n}{uC(x)}du}=\dfrac{c_nx^n}{C(x)}$. In every cases, this is a Boltzmann probability. \end{proof}

This sampler allows us to generate extremely large $(1,1)$-cycles (more than 1000000 beads). Indeed, this sampler is clearly linear in the size of the output. The following figures (Fig.\ref{fig2}) show a $(1,1)$-cycle of size 100 (Fig. \ref{fig2a}) (only 100 beads for the legibility.) It also shows that we can compose our sampler with the classical builder $(*,+,Seq,Cyc,MSet,PSet,...)$. Figure \ref{fig2b} shows a random generated a necklace of $(1,1)$-cycles. We can see that the necklaces do not contain a lot of $(1,1)$-cycles. Moreover only one of these $(1,1)$-cycles contains a lot of beads. 

\begin{figure}
  \begin{center}
    \subfigure[random (1,1)-cycle of size 100.\label{fig2a}]{\includegraphics[width=1.6in]{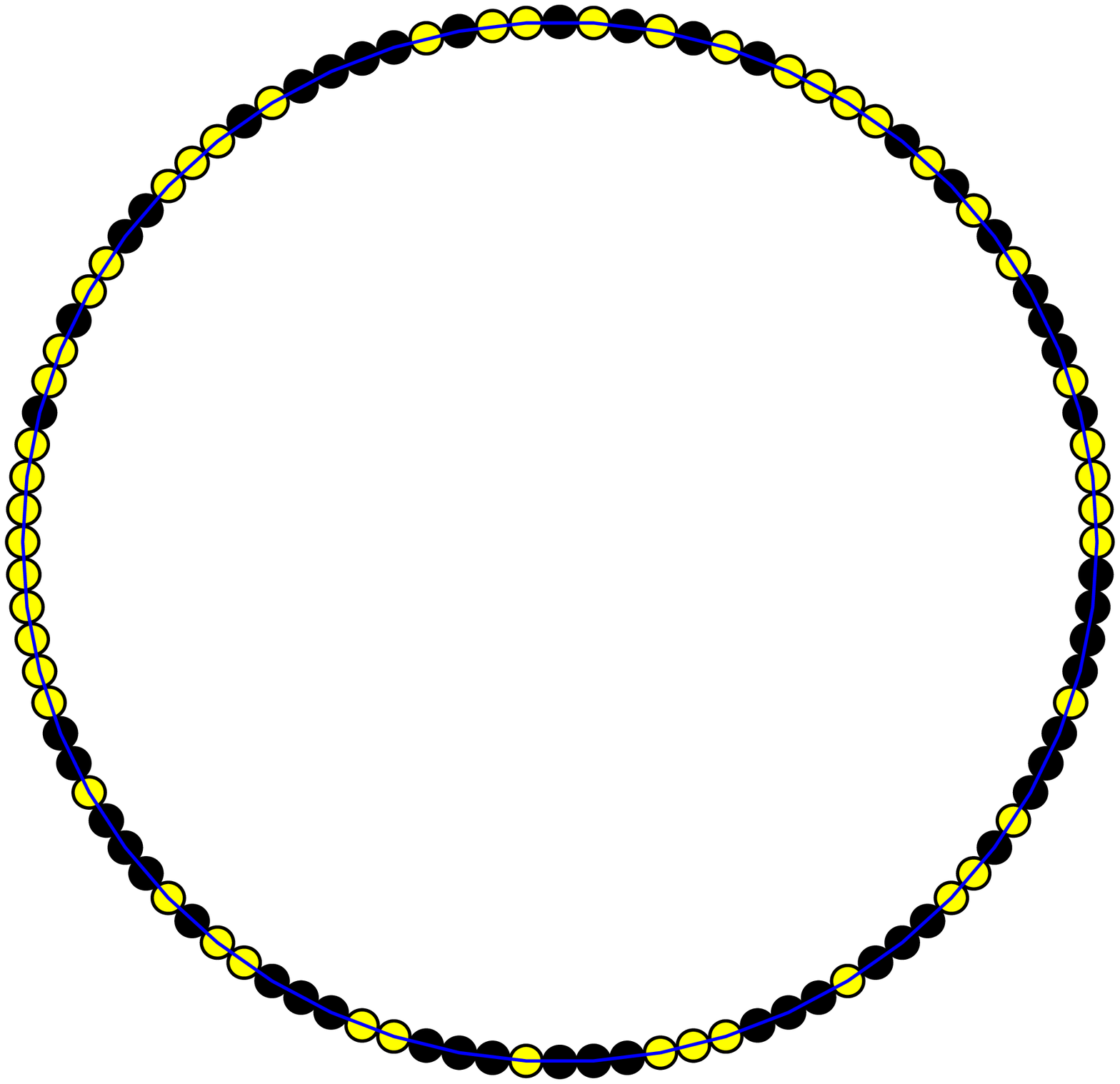}
}
     \qquad%
    \subfigure[A random necklace of (1,1)-cycles of size 200.
    \label{fig2b}]{\includegraphics[width=1.6in]{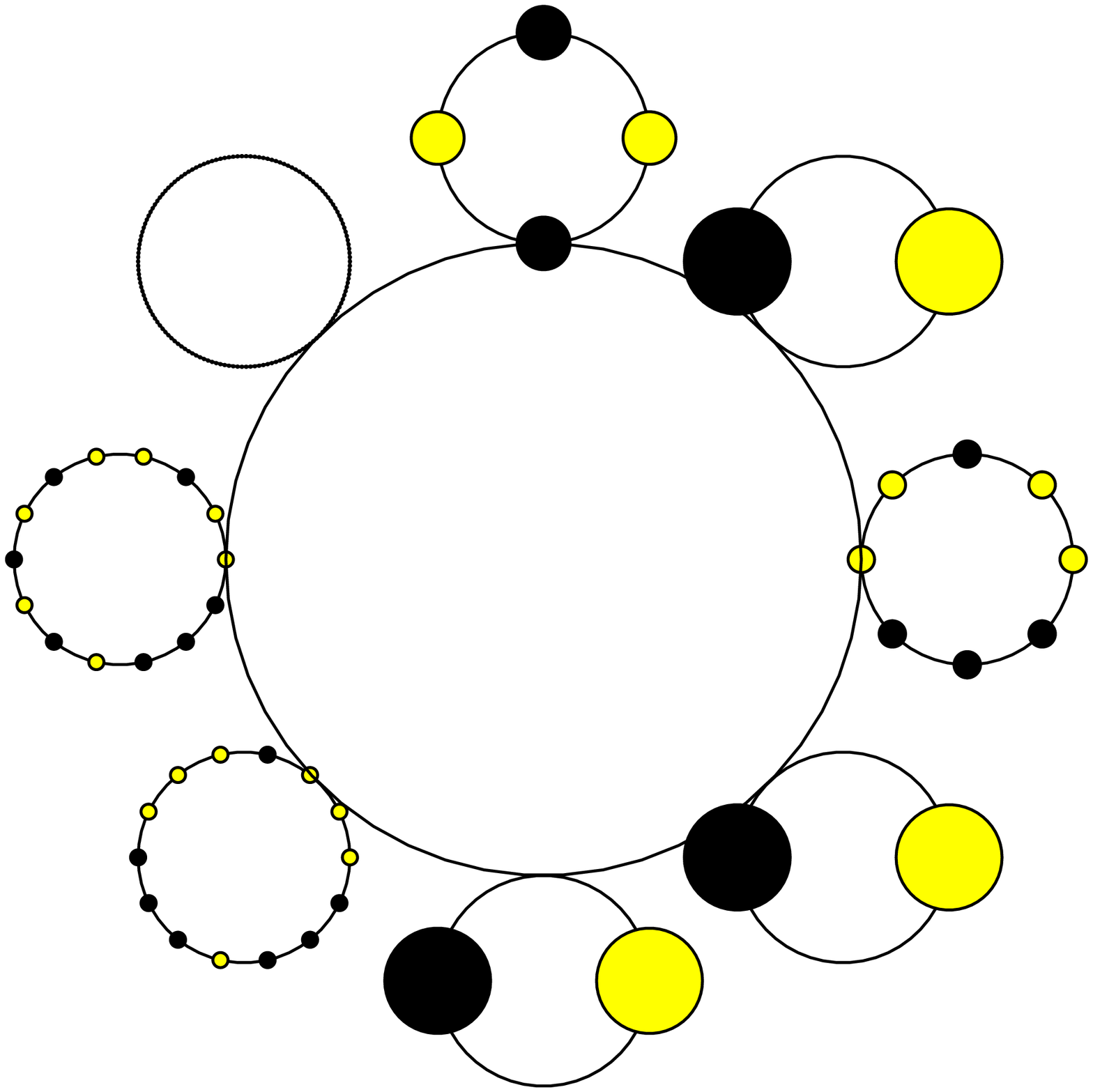}
}
    \caption{Examples of Boltzmann sampling.}
  \label{fig2}
  \end{center}
\end{figure}

\section{The general vectorial case} 

In this part, we extend the previous method to all cases, algebraic or not.
Let ${v}=(v_1, v_2, ... v_m)$, and $|v| =\sum\limits_{i=1}^{m} v_i$. Our goal here is to generate $Cyc_v$ : cycles of $m$ colors, such that the number of occurences of each atom $\mathcal{Z}_i$ verify the $v$-balance condition. We follow the same principles than in section 3.2 but the proofs are slightly more technical.

\subsection{Sampler for $Seq_{v}$}

Let us recall  $S_{v}$, the generating function of ${v}$-balanced sequences:\\
 $~~~~~~~~~~~~~~~~~~~~~~~~~~~~~~~~~~~~~~\displaystyle{ S_{{v}}=\sum\limits_{p>0}\frac{(|v|p)!}{\prod\limits_{i=1}^{m}((v_ip)!)}\prod\limits_{i=1}^{m}Z_i^{v_ip}}.$ 


\begin{algorithm}[H] 
\caption{$\Gamma_x Seq_{v}$}

\KwIn{the parameter $t$}

\KwOut{a $v$-balanced sequence of $\mathcal{Z}_i$.}

Let $L$ be a random variable in $\mathbb{N}^*$ verifying 
$\ds \mathbb{P}(L= l)=\frac{(|v|l)!}{\prod\limits_{i=1}^{m}(v_il!)} \frac{t^{(\sum\limits_{i=1}^{m}v_il)}}{S_{{v}}}$\\
Draw $l$ according to the law of $L$.\\
Let $M$ be a $(|v|l)$-uple,\\
\For{$i$ \emph{\textbf{from }} $1$ \textbf{\emph{to }} $m$}{
Select uniformly $lv_i$ positions belongs the $(|v|l)$ entries not yet affected in $M$.\\ This positions are the $\mathcal{Z}_i$ entries of $M$.
}
%
%
%
%

\Return $(M)$.
\end{algorithm}
\begin{lemma}
 Algorithm 6 is a valid Boltzmann sampler for $Seq_v$. Its arithmetic complexity is linear in the size of its output object.
\end{lemma}
\begin{proof}The proof can be easily transposed from $(1,1)$-balanced one. The complexity result is trivial.\end{proof}










\smallskip

Now, as with the example of $(1,1)$-balanced cycles, we are going to use this sampler for $v$-balanced sequences to generate $v$-balanced cycles.

\smallskip

\subsection{An isomorphism for $\Theta Cyc_v$}
\begin{theorem}
 $\Theta Cyc_{v}$ is isomorphic to $\sum\limits_{n>0} \varphi(n) Rep_n(Seq_{v})$.
\end{theorem}

\begin{proof}
Let $C_{m\mathrm{col}}$ be the generating function of cycles of $m$ atoms $\mathcal{Z}_1,..., \mathcal{Z}_m$.
$$ C_{m\mathrm{col}}=\sum\limits_{n>0}\frac{-\varphi(n)}{n}\log(1-(\sum\limits_{i=1}^{m}Z_i^n))=\sum\limits_{n>0}\frac{\varphi(n)}{n} \sum\limits_{k>0}\frac{1}{k} (\sum\limits_{i=1}^{m}Z_i^n)^k$$
$$C_{m\mathrm{col}}=\sum\limits_{n>0}\frac{\varphi(n)}{n} \sum\limits_{k>0}\frac{1}{k} \sum\limits_{\sum\limits_{i=1}^m p_i = k} \frac{k!}{\prod\limits_{i=1}^m (p_i!)}\prod\limits_{i=1}^m Z_i^{np_i} $$

\smallskip

Let $C_{v}$ be the generating function of $Cyc_v$. This is the extraction of terms with the exponents verifying the $v$-balanced condition in $C_{m\mathrm{col}}$.
$$ C_{{v}}=\sum\limits_{l>0}([\prod\limits_{i=1}^mZ_i^{lv_i}] C_{m\mathrm{col}})X^{|v|l}=\sum\limits_{n>0}\frac{\varphi(n)}{n}\sum\limits_{p>0}\frac{1}{|v|p} \frac{(|v|p)!}{\prod\limits_{i=1}^{m}((v_ip)!)}X^{|v|np} $$

We will now apply the same idea that we described for the $(1,1)$-balanced case to $\Theta \mathcal{C}_v$ (the generating function of which is $ C_{{v}}^{\bullet}$).
$$ C_{{v}}^{\bullet}=\sum\limits_{n>0} \varphi(n) \sum\limits_{p>0}\frac{(|v|p)!}{\prod\limits_ {i=1}^{m}((v_ip)!)}X^{|v|np} = \sum\limits_{n>0} \varphi(n) S_{{v}}(X^n) $$
\end{proof}

This isomorphism can be used to obtain the following sampler:
\bigskip

\begin{algorithm}[H]
\caption{$\Gamma_x \Theta Cyc_v$}

\KwIn{the parameter $x$}

\KwOut{a $v$-balanced cycle of $\mathcal{Z}_i$.}

$n:=1$\\
$S:=\frac{\varphi(n) S_{v}(x^{n}) } {C_{v}^{\bullet}(x)}$\\
Draw a real number $u$ uniformly in $[0,1]$\\
\While{$u>S$}{
$n:=n+1$\\
$S:=S+\frac{\varphi(n) S_{v}(x^{n}) } {C_{v}^{\bullet}(x)}$
}
\Return $[\Gamma_x (Rep_n(Seq_{(1,1)}))]$\\
\end{algorithm}

%
%
%

\begin{proposition}
 Algorithm 7 is a valid Boltzmann sampler for $\Theta Cyc_v$. Its arithmetic complexity is linear in the size of its output object.
\end{proposition}
\begin{proof}
Let us consider the generation of a pointing cycle $c$.\\
It can be written as $c=u^p$, where $u$ a primitive sequence (\textit{i.e.} without replication) and $p$ is the primitive repetition order of $c$ ($|c|=p|u|$). There are $s=|u|$ shifts of $u$ which produce equivalent cycles.\\
In Algorithm 7, the generating sequence is not necessarily primitive. So, $c$ can be drawn as any $\tilde{u}^d$, with $d|p$ and $\tilde{u}$ a shift of $\frac{p}{d}$ repetitions of $u$. 

So, the probability to draw $u$ is the sum for all $d|p$ of the probability to draw $d$ as repetition order and to draw one of the $s$ shift of $\tilde{u}$ as motif :
 

$$\mathbb{P}(c)= \sum\limits_{d|p}\frac{\varphi(d) Rep_d(S_{{v}})}{C_{{v}}^{\bullet}}.\frac{s.(x^{d})^{\frac{p}{d}}}{Rep_d(S_{{v}})}=s \frac{x^{|c|}}{C_{{v}}^{\bullet}}\sum\limits_{d|p}\varphi(d) = |c| \frac{x^{|c|}}{C_{{v}}^{\bullet}}.$$
$\mathbb{P}(c,i)= \frac{1}{|c|}\mathbb{P}(c) = \frac{x^{|c|}}{C_{{v}}^{\bullet}}$ with $i \in [|1,l|],$ the choice of $i$ corresponding to the pointed atom.
The complexity ensues from the results on Boltzmann sampling.
\end{proof}

To obtain a $v$-balanced cycle from an object of $\Theta Cyc_v$, we can now apply the general algorithm 5 to $\Theta Cyc_v$. As proven previously, this provides us a Boltzmann sampler for $Cyc_v$.

\section{Conclusion}
The random generation of constrained-colored structures is in general very difficult. In a previous paper \cite{BJ}, we already investigated the generation of the $k$-colored structures and size-colored structures. In this short paper, we have presented a way to efficiently generate $v$-balanced cycles. It is possible with our samplers to generate $v$-balanced cycles of sizes reaching up to one million. Nevertheless our methods can not be directly generalized to other balanced structures. For instance, we do not know how to generate $(1,1)$-balanced general non planar (unlabelled) trees where general non planar trees can be specified as $\mathcal{T}=\mathcal{Z}.MSet(\mathcal{T})$. 
This problem is a work in progress and should be solved by a method involving multivariate Boltzmann sampler. 
Another perspective is the generation of semi-labelled structures. In these structures each atom can take a color in $\{1,...,k\}$ but, if we have an atom of color $k>1$, we need to also have at least one atom of color $k-1$. Semi-labelling is a new interesting labelling which is in a sense between unlabelling and labelling. But very little is as of yet known about it.

Finally, we want to thank, B. Salvy for his precious knowledge on hypergeometric functions and Joachim Dehais, J\'er\'emie Lumbroso and Yann Ponty for their careful reading of a preliminary version of the manuscript.

\end{document}